\documentclass[a4paper,10pt]{article}
\usepackage[utf8]{inputenc}
\usepackage{amsthm,amssymb,amsmath}
\usepackage{subcaption}
\usepackage{todonotes}
\usepackage[pdfauthor={Author's name},pdftitle={Document Title},pagebackref=true,pdftex]{hyperref}

\usetikzlibrary{shapes,snakes}

\newtheorem{theorem}{Theorem}
\newtheorem{lemma}[theorem]{Lemma}
\newtheorem{corollary}[theorem]{Corollary}
\newtheorem{claim}[theorem]{Claim}
\newtheorem{proposition}[theorem]{Proposition}
\newtheorem{observation}[theorem]{Observation}

\theoremstyle{definition}
\newtheorem{example}{Example}

\newcommand{\var}{\mathsf{var}}
\newcommand{\calT}{\mathcal{T}}
\newcommand{\cc}{\mathsf{cc_{best}^{1/3}}}

\newcommand{\ccd}{\mathsf{cc}}
\newcommand{\tw}{\mathsf{tw}}
\newcommand{\twi}{\mathsf{tw_i}}

\newcommand{\twp}{\mathsf{tw_p}}
\newcommand{\twd}{\mathsf{tw_d}}
\newcommand{\wi}{\mathsf{wi}}
\newcommand{\pp}{\mathsf{p}}
\newcommand{\scw}{\mathsf{scw}}
\newcommand{\mtw}{\mathsf{mtw}}
\newcommand{\mimw}{\mathsf{mimw}}
\newcommand{\cw}{\mathsf{cw}}

\newcommand{\calC}{\mathcal{C}}

\newcommand{\constraint}[1]{\textbf{#1}}
\newcommand{\perm}{\constraint{PERM}$_n$}
\newcommand{\atmostone}{\constraint{AtMostOne}$_n$}

\title{Graph Width Measures for CNF-Encodings with Auxiliary Variables}
\author{Stefan Mengel\thanks{CRIL, CNRS} \and Romain Wallon\thanks{CRIL, Univ. Artois and CNRS}}

\begin{document}

\maketitle

\begin{abstract}
We consider bounded width CNF-formulas where the width is measured by popular graph width measures on graphs associated to CNF-formulas. Such restricted graph classes, in particular those of bounded treewidth, have been extensively studied for their uses in the design of algorithms for various computational problems on CNF-formulas. Here we consider the expressivity of these formulas in the model of clausal encodings with auxiliary variables. We first show that bounding the width for many of the measures from the literature leads to a dramatic loss of expressivity, restricting the formulas to those of low communication complexity. We then show that the width of optimal encodings with respect to different measures is strongly linked: there are two classes of width measures, one containing primal treewidth and the other incidence cliquewidth, such that in each class the width of optimal encodings only differs by constant factors. Moreover, between the two classes the width differs at most by a factor logarithmic in the number of variables. Both these results are in stark contrast to the setting without auxiliary variables where all width measures we consider here differ by more than constant factors and in many cases even by linear factors.
\end{abstract}

\section{Introduction}

Graph width measures like treewidth and cliquewidth have been studied extensively in the context of propositional satisfiability. The general idea is to assign graphs to CNF-formulas and compute their width with respect to different width measures. Then, if the resulting width is small, there are algorithms that solve SAT, but also more complex problems like \#SAT or MAX-SAT or even QBF efficiently, see e.g.~\cite{SamerS10,FischerMR08,SlivovskyS13,PaulusmaSS16,SaetherTV15,Chen04} for this line of work. There is also a considerable body of work on reasoning problems from artificial intelligence restricted to knowledge encoded by CNF-formulas with restricted underlying graphs: for example, treewidth restrictions have been studied for abduction, closed world reasoning, circumscription, disjunctive logic programming~\cite{GottlobPW10} and answer set programming~\cite{JaklPW09}.
There is thus by now a large body of work on how problems can be solved on bounded width CNF-formulas for different graph width measures.

Curiously, however, there seems to be very little work on the natural question of what we can actually encode with these restricted CNF-formulas. This question is pertinent because good algorithms for problems are less attractive if they cannot deal with interesting instances.  
We make two main contributions on the expressivity of bounded width CNF-formulas here.

As a first main contribution, we show, for a wide class of width measures, that one can give width lower bounds of any encoding of a function by means of communication complexity (Theorem~\ref{thm:abstract}). Such lower bounds were known for tree\-width~\cite{BriquelKM11}, but with our general approach, we extend them for many different width measures, in particular (signed and unsigned) cliquewidth~\cite{FischerMR08,SlivovskyS13}, modular treewidth~\cite{PaulusmaSS16} and MIM-width~\cite{SaetherTV15}. As a consequence, in a sense, for all these measures, formulas of bounded width can only encode simple functions.

All these lower bounds not only work for \emph{representations} of functions as CNF-formulas but also on \emph{clausal encodings}, i.e.~CNF-formulas using auxiliary variables. It is folklore that adding auxiliary variables can decrease the size of an encoding: for example the parity function has no subexponential CNF-representations but there is an easy linear size encoding using auxiliary variables. We here observe a similar effect for the example of treewidth: we show that any CNF-representation of the \atmostone-function of $n$ inputs without auxiliary variables has primal treewidth $n-1$ which is the highest possible. But when authorizing the use of auxiliary variables, \atmostone{} can be computed with formulas of bounded treewidth easily. This shows that lower bounds for clausal encodings are far stronger than those of CNF-representations. Considering that \atmostone{} is arguably a very easy function, we feel that encodings with auxiliary variables are the more interesting notion in our setting so we focus on them here.

We remark that this is of course not the first time that communication complexity has been used to show lower bounds on the size or width of representations for Boolean functions. In fact, this is one of the motivations for the development of the area and there is a large literature on this, see e.g.~the textbooks~\cite{KushilevitzN97,DBLP:books/daglib/0087324,DBLP:books/daglib/0028687}. In particular, there are many results for showing lower bounds on different forms of branching programs by means of communication complexity, see e.g.~\cite{Wegener00,DBLP:journals/iandc/DurisHJSS04}. More recently, this approach has been generalized to more general languages considered in knowledge compilation~\cite{PipatsrisawatD10,BovaCMS16}. However, beyond the already discussed lower bounds on treewidth in~\cite{BriquelKM11}, we are not aware of any use of communication complexity to show bounds on width measures of CNF-formulas.

In a second main contribution, we focus on the \emph{relative} expressive power of different graph width measures for clausal encodings. For the graph width measures studied in the literature, it is known that without auxiliary variables the expressivity of bounded width CNF-formulas is different for all notions and they form a partial order with so-called MIM-width as the most general notion, see e.g.~\cite[Section 5]{Brault-BaronCM14}. Somewhat surprisingly, the situation changes completely when one allows auxiliary variables: in this setting, the commonly considered width notions are all up to constant factors equivalent to either primal treewidth or to incidence cliquewidth (Theorem~\ref{thm:main}). This is true for every individual function. 
We remark that for the parameters primal treewidth, dual treewidth and incidence treewidth, it was already known that the width of encodings minimizing the respective width measures differs only by constant factors~\cite{SamerS10b,BriquelKM11,LampisMM18}. All other relationships are new.

We also show that, assuming that an optimal encoding of a function has at least primal treewidth $\log(n)$ where $n$ is the number of variables, incidence cliquewidth and primal treewidth differ exactly by a factor of $\Theta(\log(n))$ for optimal encodings. So, up to a logarithmic scaling, in fact all the width measures in~\cite{SamerS10,FischerMR08,SlivovskyS13,PaulusmaSS16,SaetherTV15} coincide when allowing auxiliary variables. Note that this scaling exactly corresponds to the runtime differences of many algorithms: while treewidth based algorithms often have runtimes of the type $2^{O(k)} n^c$ for treewidth~$k$ and a constant $c$, cliquewidth based algorithms typically give runtimes roughly $n^{O(k')}$ for cliquewidth $k'$. These runtimes coincide exactly when treewidth and cliquewidth differ by a logarithmic factor which, as we show here, they do generally for encodings with auxiliary variables. 

We finally use our main results for several applications. In particular, we answer an open question of~\cite{BriquelKM11} on the cliquewidth of the permutation function \perm{}
and generalize a classical theorem on planar circuits from~\cite{LiptonT80},
see Section~\ref{sec:examples} for details.

Most of our results use machinery recently developed in the area of knowledge compilation. In particular, we use a combination of the algorithm in~\cite{BovaCMS15}, the width notion for DNNF developed in~\cite{CapelliM18} and the lower bound techniques from~\cite{PipatsrisawatD10,BovaCMS16}. Relying on these building blocks, most of our proofs become rather simple.

\section{Preliminaries}

\subsection{CNF-Formulas and their Graphs}

We use standard notations for CNF-formulas as it can e.g.~be found in~\cite{Handbook09}.
Let~$X$ be a set of variables. A \emph{CNF-representation} of a Boolean function $f$ in variables~$X$ is a CNF-formula $F$ on the variable set~$X$ that has as models exactly the assignments on which $f$ evaluates to true. A \emph{clausal encoding} of $f$ is a CNF-formula $F'$ on a variable set $X\cup Y$ such that 
\begin{itemize}
 \item for every assignment $a:X\rightarrow \{0,1\}$ on which $f$ evaluates to true, there is an extension $a'$ of $a$ to $Y$ that is a model of $F'$, and 
 \item for every assignment $a:X\rightarrow \{0,1\}$ on which $f$ evaluates to false, no extension $a'$ of $a$  to $Y$ is a model of $F'$.
\end{itemize}
The variables in $Y$ are called \emph{auxiliary variables}. An auxiliary variable $y$ is called \emph{dependent} if and only if in the first item above all extensions $a'$ satisfying $F'$ take the same value on $y$~\cite{GiunchigliaMT02}. We say that a clausal encoding has \emph{dependent auxiliary variables} if all its auxiliary variables are dependent. Note that for such an encoding the extension $a'$ is unique.

We use standard notations from graph theory and assume the reader to have a basic background in the area~\cite{Diestel12}. By $N(v)$ we denote the open neighborhood of a vertex in a graph.

In some parts of this paper, we will also deal with Boolean circuits. We assume that the reader is familiar with basic definitions in the area. As it is common when considering circuits with structurally restricted underlying graphs, we assume that every input variable appears in only one input gate. This property is sometimes called the \emph{read-once property}.

To every CNF-formula $F$, we assign two graphs. The \emph{primal graph} of $F$ has as vertices the variables of $F$ and two variables $x,y$ are connected by an edge if and only if there is a clause $C$ such that a literal in $x$ and a literal in $y$ appear in $C$. The \emph{incidence graph} of $F$ has as vertex set the union of the variable set and the clause set of $F$. Edges in the incidence graph are exactly the pairs $x,C$ where $x$ is a variable and $C$ a clause that contains a literal in $x$.

\begin{example}
\label{ex:cnf-graphs}
Let us consider the clauses $C_1 := x_1 \vee \neg x_2$, $C_2 := x_2 \vee x_3 \vee \neg x_4 \vee \neg x_5$, $C_3 := \neg x_4 \vee x_5$ and
$C_4 := x_4 \vee x_5$, and let the CNF-formula $F$ be defined as $F := C_1 \wedge C_2 \wedge C_3 \wedge C_4$.
Its primal and incidence graphs are given in Figure~\ref{ex:fig:cnf-graphs}.

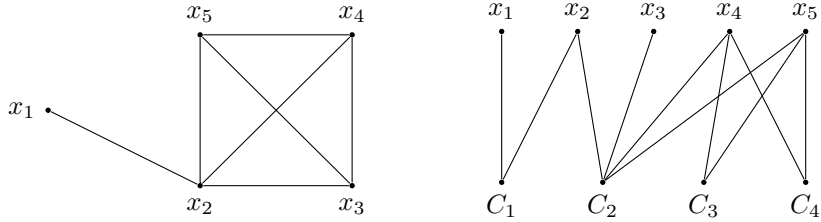
\begin{figure}
\begin{subfigure}{0.5\textwidth}
\centering
\begin{tikzpicture}

	\draw (0, 1) node (x1) [label=left:{$x_1$}, circle, fill=black, minimum size=2pt, inner sep=0pt] {};
	\draw (2, 0) node (x2) [label=below:{$x_2$}, circle, fill=black, minimum size=2pt, inner sep=0pt] {};
	\draw (4, 0) node (x3) [label=below:{$x_3$}, circle, fill=black, minimum size=2pt, inner sep=0pt] {};
	\draw (4, 2) node (x4) [label=above:{$x_4$}, circle, fill=black, minimum size=2pt, inner sep=0pt] {};
	\draw (2, 2) node (x5) [label=above:{$x_5$}, circle, fill=black, minimum size=2pt, inner sep=0pt] {};
	
	\draw (x1) edge (x2);
	
	\draw (x2) edge (x2);
	\draw (x2) edge (x3);
	\draw (x2) edge (x4);
	\draw (x2) edge (x5);
	
	\draw (x3) edge (x4);
	\draw (x3) edge (x5);
	
	\draw (x4) edge (x5);

\end{tikzpicture}
\end{subfigure}
\hfill
\begin{subfigure}{0.5\textwidth}
\centering
\begin{tikzpicture}

	\draw (0, 2) node (x1) [label=above:{$x_1$}, circle, fill=black, minimum size=2pt, inner sep=0pt] {};
	\draw (1, 2) node (x2) [label=above:{$x_2$}, circle, fill=black, minimum size=2pt, inner sep=0pt] {};
	\draw (2, 2) node (x3) [label=above:{$x_3$}, circle, fill=black, minimum size=2pt, inner sep=0pt] {};
	\draw (3, 2) node (x4) [label=above:{$x_4$}, circle, fill=black, minimum size=2pt, inner sep=0pt] {};
	\draw (4, 2) node (x5) [label=above:{$x_5$}, circle, fill=black, minimum size=2pt, inner sep=0pt] {};
	
	\draw (0,    0) node (C1) [label=below:{$C_1$}, circle, fill=black, minimum size=2pt, inner sep=0pt] {};
	\draw (1.33, 0) node (C2) [label=below:{$C_2$}, circle, fill=black, minimum size=2pt, inner sep=0pt] {};
	\draw (2.66, 0) node (C3) [label=below:{$C_3$}, circle, fill=black, minimum size=2pt, inner sep=0pt] {};
	\draw (4,    0) node (C4) [label=below:{$C_4$}, circle, fill=black, minimum size=2pt, inner sep=0pt] {};
	
	\draw (x1) edge (C1);
	\draw (x2) edge (C1);
	
	\draw (x2) edge (C2);
	\draw (x3) edge (C2);
	\draw (x4) edge (C2);
	\draw (x5) edge (C2);
	
	\draw (x4) edge (C3);
	\draw (x5) edge (C3);
	
	\draw (x4) edge (C4);
	\draw (x5) edge (C4);
\end{tikzpicture}
\end{subfigure}
\caption{Graphs associated to the CNF-formula $F$ in Example~\ref{ex:cnf-graphs}: primal graph (left) and incidence graph (right).}
\label{ex:fig:cnf-graphs}
\end{figure}
\end{example}

\subsection{Graph Width Measures} 

In this section, we will introduce several graph width measures we will consider throughout this paper.
A \emph{tree decomposition} $(T,(B_t)_{t\in V(T)})$ of a graph $G=(V,E)$ consists of a tree~$T$ and, for every node $t$ of $T$, a set $B_t\subseteq V$ called \emph{bag} such that:
\begin{itemize}
 \item $\bigcup_{t\in V(T)} B_t = V$,
 \item for every edge $uv\in E$, there is a bag $B_t$ such that $\{u,v\}\subseteq B_t$, and
 \item for every $v\in V$, the set $\{t\in V(T)\mid v\in B_t\}$ is connected in $T$.
\end{itemize}
The \emph{width} of a tree decomposition is defined as~$\max \{|B_t|\mid t\in V(T)\} -1$. The \emph{treewidth} $\tw(G)$ of $G$ is defined as the minimum width taken over all tree decompositions of $G$. The \emph{primal treewidth} $\twp(F)$ of a CNF-formula $F$ is defined as the treewidth of its primal graph and the \emph{incidence treewidth} $\twi(F)$ of $F$ is defined as that of the incidence graph.

\begin{example}
\label{ex:tw}
Let us again consider the formula $F$ of Example~\ref{ex:cnf-graphs}.
Figure~\ref{ex:fig:tw} shows a tree decomposition of the primal graph and the incidence graph of $F$.
Both of these decompositions are optimal: it is well-known that for every tree decomposition of a graph $G$, the vertices of every clique must be contained in a common bag. So, in this case, $x_2, x_3, x_4, x_5$ must be in one bag for every tree decomposition of the primal graph of $F$ and thus $\twp(F) \ge 3$ which shows that the decomposition of Figure~\ref{ex:fig:tw} is optimal and $\twp(F) = 3$. Concerning the treewidth of the incidence graph, remark that this graph has a cycle and is thus not a tree. Since trees are well-known to be the only graphs of treewidth $1$, it follows that $\twi(F) \ge 2$ and thus the decomposition in Figure~\ref{ex:fig:tw} is optimal and $\twi(F)=2$.

\begin{figure}
\centering
\begin{subfigure}{0.45\textwidth}
\centering
\begin{tikzpicture}[inner/.style={draw,inner sep=3pt, rounded corners}]
\node[inner] {$x_1, x_2$} 
	child {
		node[inner] {$x_2$, $x_3$, $x_4$, $x_5$}
	 };

\end{tikzpicture}
\end{subfigure}
\hfill
\begin{subfigure}{0.45\textwidth}
\centering
\begin{tikzpicture}[inner/.style={draw,inner sep=3pt, rounded corners},level 3/.style={sibling distance=5em}, level 4/.style={sibling distance=5em},level distance=3.5em]
\node[inner] {$x_1, C_1$} 
	child {
		node[inner] {$x_2, C_1$}
		child {
			node[inner] {$x_2, C_2$}
			child {
				node[inner] {$x_3, C_2$}
			}
			child {
				node[inner] {$x_4, x_5, C_2$}
				child {
					node[inner] {$x_4, x_5, C_3$}			
				}
				child {
					node[inner] {$x_4, x_5, C_4$}					
				}	
			}
		}
	 };
\end{tikzpicture}
\end{subfigure}
\caption{Tree decompositions of the graphs associated to the CNF-formula $F$ in Example~\ref{ex:tw}: primal graph (left) and incidence graph (right).}
\label{ex:fig:tw}
\end{figure}
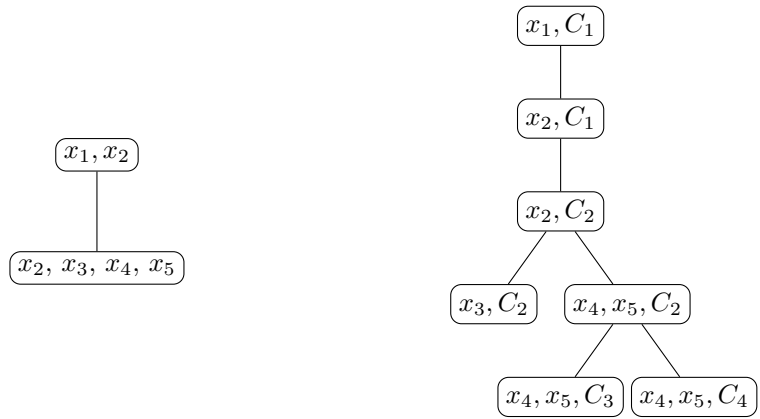
\end{example}

We say that two vertices $u$, $v$ in a graph $G=(V,E)$ have the same \emph{neighborhood type} if and only if $N(u)\setminus \{v\} = N(v)\setminus \{u\}$. It can be shown that having the same neighborhood type is an equivalence relation on $V$. 
A generalization of treewidth is \emph{modular treewidth} which is defined as follows: from a graph $G$ we construct a new graph $G'$ by contracting all vertices sharing a neighborhood type, i.e., from every equivalence class we delete all vertices but one. The modular treewidth of $G$ is then defined to be the treewidth of $G'$. The modular treewidth $\mtw(F)$ of a CNF-formula $F$ is defined as the modular treewidth of its incidence graph.

\begin{example}
\label{ex:mtw}
Let us consider again the formula $F$ from Example~\ref{ex:cnf-graphs}.
Figure~\ref{ex:fig:mtw} shows a contraction of all vertices sharing a neighborhood type in the incidence graph of $F$.
This contraction resulting in a tree, we have that $\mtw(F) = 1$.

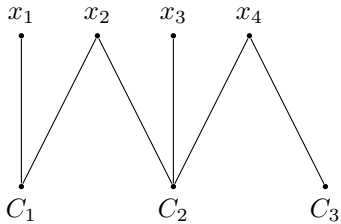
\begin{figure}
\centering
\begin{tikzpicture}[inner/.style={draw,inner sep=3pt, rounded corners}]
\draw (0, 2) node (x1) [label=above:{$x_1$}, circle, fill=black, minimum size=2pt, inner sep=0pt] {};
	\draw (1, 2) node (x2) [label=above:{$x_2$}, circle, fill=black, minimum size=2pt, inner sep=0pt] {};
	\draw (2, 2) node (x3) [label=above:{$x_3$}, circle, fill=black, minimum size=2pt, inner sep=0pt] {};
	\draw (3, 2) node (x45) [label=above:{$x_{4}$}, circle, fill=black, minimum size=2pt, inner sep=0pt] {};
	
	\draw (0, 0) node (C1) [label=below:{$C_1$}, circle, fill=black, minimum size=2pt, inner sep=0pt] {};
	\draw (2, 0) node (C2) [label=below:{$C_2$}, circle, fill=black, minimum size=2pt, inner sep=0pt] {};
	\draw (4, 0) node (C34) [label=below:{$C_{3}$}, circle, fill=black, minimum size=2pt, inner sep=0pt] {};
	
	\draw (x1) edge (C1);
	\draw (x2) edge (C1);
	
	\draw (x2) edge (C2);
	\draw (x3) edge (C2);
	\draw (x45) edge (C2);
	
	\draw (x45) edge (C34);
\end{tikzpicture}
\caption{A contraction of the incidence graph of the CNF-formula $F$ in Example~\ref{ex:mtw}. 
In the original graph, $x_4$ and $x_5$ have the same neighborhood type, as do $C_3$ and $C_4$.
We thus get the shown contraction by deleting $x_5$ and $C_4$.
Note that the obtained graph is a tree.}
\label{ex:fig:mtw}
\end{figure}
\end{example}

The \emph{cliquewidth} $\cw(G)$ of a graph $G$ is defined as the minimum number of labels needed to construct $G$ with the following operations:
\begin{itemize}
 \item creating a new vertex with label $i$,
 \item taking the disjoint union of two labeled graphs,
 \item joining all vertices with a label $i$ to all vertices with a label $j$ for $i\ne j$, and
 \item renaming a label $i$ to $j$ for $i\ne j$.
\end{itemize}
The \emph{incidence cliquewidth}~$\cw(F)$ of a formula $F$ is defined as the cliquewidth of the incidence graph of $F$~\cite{SlivovskyS13}.

Finally, we consider the adaption of cliquewidth to signed graphs. To this end, let us make some additional definitions. The \emph{signed incidence graph} $G'$ of a CNF-formula $F$ is the graph we get from the incidence graph $G=(V,E)$ by labeling the edges with $\{+,-\}$ as follows: 
\begin{itemize}
 \item every edge $xC$ such that $x$ appears positively in $C$ is labeled by $+$, and 
 \item every edge $xC$ such that $x$ appears negatively in $C$ is labeled by $-$.
\end{itemize}

The signed cliquewidth of a graph $G'$ is defined as the minimum number of labels needed to construct $G'$ with the following operations:
\begin{itemize}
 \item creating a new vertex with label $i$,
 \item taking the disjoint union of two labeled graphs,
 \item joining all vertices with a label $i$ to all vertices with a label $j$ for $i\ne j$ by an edge with label $+$,
 \item joining all vertices with a label $i$ to all vertices with a label $j$ for $i\ne j$ by an edge with label $-$, and
 \item renaming a label $i$ to $j$ for $i\ne j$.
\end{itemize}
The \emph{signed incidence cliquewidth} $\scw(F)$ of $F$ is defined as the signed cliquewidth of its signed incidence graph~\cite{FischerMR08}.

We will deal with several other graph width measures for a CNF-formula in the remainder of this paper, in particular dual treewidth~$\twd(F)$ and MIM-width $\mimw(F)$. Since for those notions we will only use some of their properties, we will refrain from overwhelming the reader by giving their definitions and refer to the literature, e.g.~\cite{SamerS10,FischerMR08,vatshelleThesis,SaetherTV15,SlivovskyS13}.

We also consider the treewidth $\tw(C)$ and the cliquewidth $\cw(C)$ of Boolean circuits $C$.

\subsection{Communication Complexity}

Here we give some very basic notions of communication complexity, focusing only on so-called combinatorial rectangles, which are an important object in the field.
For more details, the reader is referred to the very readable textbook~\cite{KushilevitzN97}. 

Let $X$ be a set of variables and $\Pi = (Y,Z)$ a partition of $X$. A \emph{combinatorial rectangle} respecting $\Pi$ is a Boolean function $r(X)$ that can be written as a conjunction $r(X) = r_1(Y) \land r_2(Z)$.
For a Boolean function $f$ on $X$, a \emph{rectangle cover of size $s$} respecting $\Pi$ is defined to be a representation \[f(X) = \bigvee_{i=1}^s r^i(X) = \bigvee_{i=1}^s r_1^i(Y) \land r_2^i(Z),\] 
where all  $r^i(X) = r_1^i(Y) \land r_2^i(Z)$ are combinatorial rectangles respecting $\Pi$.
The non-deterministic communication complexity $\ccd(f,\Pi)= \ccd(f, (Y, Z))$ of~$f$ is defined as $\log(s_{\min})$ where $s_{\min}$ is the minimum size of any rectangle cover of~$f$ respecting~$\Pi$.

\begin{example}
 By definition, all formulas in disjunctive normal forms are rectangle covers of the functions they compute respecting all possible partitions. For example, \[F= (\neg x\land \neg y \land z) \lor (x\land y \land z) \lor (x \land \neg y \land \neg z)\] is a rectangle cover of size $3$ respecting every partition of $\{x,y,z\}$. However, for example for the partition $(\{x,y\}, \{z\})$, there is the smaller rectangle cover 
 \[( ((\neg x \land \neg y) \lor (x \land y)) \land z) \lor (x \land \neg y \land \neg z)\] of size $2$. It is not hard to see that there is no smaller rectangle cover of $F$ for this partition.
\end{example}

The best-case non-deterministic communication complexity with $\frac 1 3$-balance $\cc(f)$ is defined as $\cc(f) := \min_\Pi(\ccd(f, \Pi))$ where the minimum is over all partitions $\Pi=(Y,Z)$ of $X$ with $\min(|Y|, |Z|) \ge |X|/3$.

\begin{example}
 Consider the function $\constraint{EQ}_n(x_1, \ldots x_n, y_1, \ldots, y_n)$ which is true if and only if for every $i\in [n]$ we have $x_i = y_i$. It is well-known that for the partition $\Pi_1 = (\{x_1, \ldots, x_n\}, \{y_1, \ldots, y_n\})$ we have $\ccd(\constraint{EQ}_n, \Pi_1)= n$, see e.g.~\cite[Chapter 2]{KushilevitzN97}. However, for the partition \[\Pi_2 = (\{x_1, y_1, \ldots, x_{\lceil n/2 \rceil}, y_{\lceil n/2 \rceil}\}, \{x_{\lceil n/2 \rceil+1}, y_{\lceil n/2 \rceil+1}, \ldots, x_n, y_n\})\] we have that \[\constraint{EQ}_n(x_1, \ldots x_n, y_1, \ldots, y_n)= \left(\bigwedge_{i=1}^{\lceil n/2 \rceil} x_i = y_i \right) \land \left( \bigwedge_{i = \lceil n/2 \rceil}^n x_{i} = y_{i} \right)\] is a rectangle cover of size $1$ respecting $\Pi_2$. Thus, we have $\cc(\constraint{EQ}_n) = \ccd(\constraint{EQ}_n, \Pi_2) = 0$.
\end{example}

\subsection{Structured Deterministic DNNF}

Out of the rich landscape of representations from knowledge compilation, see e.g.~\cite{DarwicheM02,PipatsrisawatD08}, we only introduce one that we will use in the remainder of this paper.
For all circuits in this section, we assume that $\land$-gates have exactly two inputs while the number of $\lor$-gates may be arbitrary.

A \emph{v-tree} $T$ for a variable set $X$ is a full binary tree whose leaves are in bijection with $X$. We call the variable assigned by this bijection to a leaf $v$ the \emph{label} of $v$. For a node $t\in T$, we denote by $T_t$ the subtree of $T$ that has $t$ as its root and by $\var(T_t)$ the variables that are labels of leaves in $T_t$.

\begin{example}\label{ex:vtree}
 We give a v-tree for the variable set $\{x,y,z\}$ on the left of Figure~\ref{ex:fig:dnnf}.
\end{example}

\begin{figure}
\centering
    \begin{tikzpicture}[]
  \node[draw,circle,label=right:{$a$}] (top) at (0, 0) {};
  \node[draw,circle] (i1) at (-.5, -1) {$x$};
  \node[draw,circle] (i11) at (1, -2) {$z$};
  \node[draw,circle] (i10) at (0, -2) {$y$};
  \node[draw,circle,label=right:{$b$}] (i2) at (.5, -1) {};
%   \draw[->] (i11)--(i1);
  \draw[<-] (i10)--(i2);
  \draw[<-] (i11)--(i2);
  \draw[<-] (i1)--(top);
  \draw[<-] (i2)--(top);
  
\node[draw,circle split,scale=.6, minimum size=1.2cm] (a1) at (5,0) {$\lor$ \nodepart{lower} $a$} ;
\node[draw,circle split,scale=.6, minimum size=1.2cm] (a2) at (4.5,-1) {$\land$ \nodepart{lower} $a$} ;
\node[draw,circle split,scale=.6, minimum size=1.2cm] (a3) at (5.5,-1) {$\land$ \nodepart{lower} $a$} ;

\node[draw,circle split,scale=.6, minimum size=1.2cm] (x1) at (3,-2) {$x$ \nodepart{lower} $x$} ;
\node[draw,circle split,scale=.6, minimum size=1.2cm] (x2) at (4,-2) {$\neg x$ \nodepart{lower} $x$} ;

\node[draw,circle split,scale=.6, minimum size=1.2cm] (b1) at (6,-2) {$\lor$ \nodepart{lower} $b$} ;
\node[draw,circle split,scale=.6, minimum size=1.2cm] (b2) at (7,-2) {$\lor$ \nodepart{lower} $b$} ;

\node[draw,circle split,scale=.6, minimum size=1.2cm] (b3) at (5.5,-3) {$\land$ \nodepart{lower} $b$} ;
\node[draw,circle split,scale=.6, minimum size=1.2cm] (b4) at (6.5,-3) {$\land$ \nodepart{lower} $b$} ;
\node[draw,circle split,scale=.6, minimum size=1.2cm] (b5) at (7.5,-3) {$\land$ \nodepart{lower} $b$} ;

\node[draw,circle split,scale=.6, minimum size=1.2cm] (y1) at (5,-4) {$y$ \nodepart{lower} $y$} ;
\node[draw,circle split,scale=.6, minimum size=1.2cm] (y2) at (6,-4) {$\neg y$ \nodepart{lower} $y$} ;
\node[draw,circle split,scale=.6, minimum size=1.2cm] (z1) at (7,-4) {$z$ \nodepart{lower} $z$} ;
\node[draw,circle split,scale=.6, minimum size=1.2cm] (z2) at (8,-4) {$\neg z$ \nodepart{lower} $z$} ;

\draw[->] (a2)--(a1);
\draw[->] (a3)--(a1);
\draw[->] (x1)--(a2);
\draw[->] (x2)--(a3);
\draw[->] (b1)--(a2);
\draw[->] (b2)--(a3);
\draw[->] (b3)--(b1);
\draw[->] (b4)--(b1);
\draw[->] (b5)--(b2);
\draw[->] (y1)--(b3);
\draw[->] (y2)--(b4);
\draw[->] (y2)--(b5);
\draw[->] (z1)--(b3);
\draw[->] (z2)--(b4);
\draw[->] (z2)--(b5);

\end{tikzpicture}
\caption{A v-tree on the left and a complete structured DNNF structured by this v-tree. For the internal nodes of the v-tree, we give node names on the right of the nodes whereas for leaves we assume that the name is the label.
All gates of the complete structured DNNF show the operation of the gate (on top) and the name $t$ of the node in the v-tree for which this gate is in $\mu(t)$ (on bottom).}
\label{ex:fig:dnnf}
\end{figure}
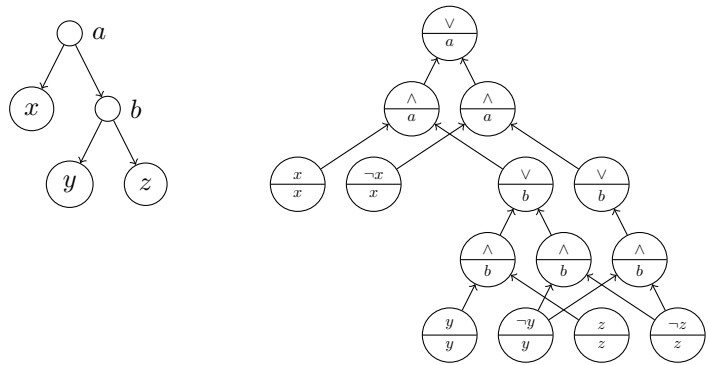

We give some definitions from~\cite{CapelliM18}.
A \emph{complete structured DNNF} $D$ structured by a v-tree $T$ is a Boolean circuit with the following properties: there is a labeling $\mu$ of the nodes in $T$ with subsets of gates of $D$ such that:
\begin{itemize}
 \item For every gate $g$ of $D$ there is a unique node $t_g$ of $T$ with $g \in \mu(t_g)$.
 \item If $t$ is a leaf labeled by a variable $x$, then $\mu(t)$ may only contain $x$ and $\neg x$. Moreover, for every input gate $g$, the node $t_g$ is a leaf.
 \item For every $\lor$-gate $g$, all inputs are $\land$-gates in $\mu(t_g)$.
 \item Every $\land$-gate $g$ has exactly two inputs $g_1, g_2$ that are both $\lor$-gates or input gates. Moreover, $t_{g_1}$ and $t_{g_2}$ are the children of $t_g$ in $T$ and in particular $t_{g_1}\ne t_{g_2}$.
\end{itemize}
The \emph{width} $\wi(D)$ of $D$ is defined as the maximal number of $\lor$-gates in any set $\mu(t)$.
We often speak of complete structured DNNF without mentioning the v-tree by which it is structured in cases where the form of the v-tree is unsubstantial.
Intuitively, a complete structured DNNF is a Boolean circuit in negation normal form in which the gates are organized into blocks $\lambda(t)$ which form a tree shape. In every block one then computes a 2-DNF whose inputs are gates from the blocks that are the children of $\lambda(t)$ in the tree shape.

\begin{example}
 On the right side of Figure~\ref{ex:fig:dnnf}, we give a complete structured DNNF structured by the v-tree of Example~\ref{ex:vtree}. There are $3$ $\lor$-gates in $\mu(b)$, so the width of the given complete structured DNNF is $3$.
\end{example}

A complete structured DNNF is called \emph{deterministic} if and only if for every assignment and for every $\lor$-gate, at most one input evaluates to true.
Note that we do not allow constant input gates here. We remark that if we allowed those, we could always get rid of them in the circuit by propagation without changing any other properties of the circuit, see~\cite[Section~4]{CapelliM18}.
We also remark that in a complete structured DNNF $D$, we can \emph{forget} a variable $x$, i.e., construct a complete structured DNNF $D'$ computing $\exists x D$, by setting all occurrences of $x$ and $\neg x$ to $1$ and propagating the constants in the obvious way. This operation does not increase the width, see~\cite{CapelliM18}. However, if $D$ is deterministic, this is generally not the case for $D'$.

\section{The Effect of Auxiliary Variables}

In this section, we will motivate the use of auxiliary variables when considering width measures of CNF-encodings. To this end, we will show with an example that auxiliary variables may arbitrarily reduce the treewidth of encodings. Note that this is not very surprising since it is not too hard to see that CNF-representations of, say, the parity function, are of high treewidth. However, in this case the \emph{size} of the representation is exponential, so in a sense parity is a hard function for CNF-representations anyway. Here we will show that even for functions that have small CNF-representations there can be a large gap between the treewidth of representations and clausal encodings with auxiliary variables. That is why we think it is useful to systematically study width measures for clausal encodings.

As an example for a function where auxiliary variables have a dramatic impact on width, consider the \atmostone-function on variables $x_1, \ldots, x_n$ which accepts exactly those assignments in which at most one variable is assigned to $1$. There is an obvious quadratic size representation as
\[\text{\atmostone}(x_1,\dots,x_n) = \bigwedge_{i,j\in [n], i < j} \neg x_i \lor \neg x_j.\]
However, this representation has as primal graph the clique $K_n$ which is of treewidth $n-1$. We will see that in fact there is no representation of \atmostone{} that is of smaller primal treewidth unless one adds auxiliary variables, in which case there is a simple encoding of primal treewidth~$2$.

\begin{theorem}
\label{thm:amo}
Any CNF-representation of the \atmostone-function of $n$ inputs without auxiliary variables has primal treewidth $n-1$.
However, there is a clausal encoding of \atmostone{} of primal treewidth~$2$.
\end{theorem}

To prove Theorem~\ref{thm:amo}, we split the statement into two lemmas.

\begin{lemma}
Any CNF-representation of the \atmostone-function of $n$ inputs without auxiliary variables has primal treewidth $n-1$.
\end{lemma}

\begin{proof}
Let $x_1,\dots,x_n$ be the variables of \atmostone. We proceed with two claims.

\begin{claim}
\label{atleasttwoneg}
Every non-tautological clause $C$ of any CNF-representation of \linebreak\atmostone{} must contain at least
the negation of two variables from $x_1,\dots,x_n$.
\end{claim}

\begin{proof}
Suppose that a clause $C$ does not contain two such literals.
Then, there are two possible cases: either $C$ contains no negated variables, or exactly one.
In the first case, the model of \atmostone{} setting all variables to $0$ does not satisfy $C$,
so $C$ cannot be part of the CNF-representation.
In the second case, let $x_i$ be the (only) variable of \atmostone{} appearing negatively in $C$.
Then, the model of \atmostone{} setting only $x_i$ to $1$ and all other variables to $0$ does
not satisfy $C$, so $C$ cannot be part of the CNF-representation, either.
Hence, at least two negated variables must appear in $C$. 
\end{proof}

% arxiv add labels to claims to better reference them
From Claim~\ref{atleasttwoneg}, we will deduce that all pairs of variables must appear conjointly in at least one
clause.

\begin{claim}\label{clm:pairs}
\label{allcouples}
For each pair of variables $x_i, x_j$ from $x_1,\dots,x_n$ with $i \neq j$, there is a clause
in the CNF-representation of \atmostone{} containing both $\neg x_i$ and $\neg x_j$.
\end{claim}

\begin{proof}
Suppose that, for a pair $x_i, x_j$, such a clause does not exist. Let $a$ be the assignment that sets exactly the variables $x_i, x_j$ to $1$ and all other variables to~$0$.
Let $C$ be a clause from the CNF-representation.
By our previous claim, $C$ contains two negated variables from $x_1,\dots,x_n$.
Because of our assumption, at least one of these literals is neither $\neg x_i$ nor
$\neg x_j$, and this literal is satisfied by $a$. Thus $C$ is satisfied by $a$.
Since this is true for every clause $C$, it follows that $a$ satisfies all the clauses of
the representation, so it is one of its models.
However, $a$ is not a model of \atmostone.
As a consequence, a clause containing both $\neg x_i$ and $\neg x_j$ must exist,
which is also true for every pair $x_i, x_j$.
\end{proof}

Claim~\ref{allcouples} shows that for each pair of variables, there is a clause containing
both of them.
It follows that all variables are connected to all other variables in the primal graph of the representation. So the primal graph is a clique which has treewidth $n-1$.
\end{proof}

We now prove the second part of Theorem~\ref{thm:amo}, which shows that if we allow the use of auxiliary variables, we may decrease the treewidth dramatically.

\begin{lemma}
 There is a clausal encoding of \atmostone{} of primal treewidth $2$.
\end{lemma}
\begin{proof}
 We use the well-known ladder encoding from~\cite{GentN04}, see also~\cite[Section 2.2.5]{Handbook09}. We introduce the auxiliary variables $y_0, \ldots, y_n$. The encoding consists of the following clauses, for every $i\in [n]$. :
 \begin{itemize}
  \item the validity clauses $\neg y_{i-1} \lor y_i$, and 
  \item clauses representing the constraint $x_i \leftrightarrow (\neg y_{i-1} \land y_i)$ 
 \end{itemize}
It is easy to see that this encoding is correct: the auxiliary variables $y_i$ encode if one of the variables $x_j$ for $j \leq i$ is assigned to $1$.
Concerning the treewidth bound, we construct for every index $i\in [n]$ the bag $B_i:=\{y_{i-1}, y_i, x_i\}$. Then $(P_n, (B_i)_{i\in [n]})$ where $P_n$ has nodes $[n]$ and edges $\{(i,i+1) \mid i \in [n-1]\}$ is a tree decomposition of the encoding of width $2$.
\end{proof}
% \todo{discuss?}

\section{Width vs. Communication}
\label{sec:wicomm}

In this section, we show that from communication complexity we get lower bounds for the various width notions of Boolean functions. The main building block is the following result that is an application of the main result of~\cite{PipatsrisawatD10} to complete structured DNNF.

\begin{theorem}\label{thm:ijcai}
 Let $D$ be a complete structured DNNF structured by a v-tree $T$ computing a function $f$ in variables $X$. Let $t$ be a node of $T$ and let $Y:= \var(T_t)$ and $Z=X\setminus \var(T_t)$. Finally, let $\ell$ be the number of $\lor$-gates in $\mu(t)$. Then there is a rectangle cover of $f$ respecting $(Y,Z)$ of size at most $\ell$.
\end{theorem}

Note that in~\cite{PipatsrisawatD10} the considered models are structured DNNF that are not necessarily complete, a slightly more general model than ours. Thus the statement in~\cite{PipatsrisawatD10} is slightly different. However, it is easy to see that in our restricted setting, their proof shows the statement we give above, see also the discussion in~\cite[Section 5]{BovaCMS16}.
Since Theorem~\ref{thm:ijcai} is somewhat technical, it will be more convenient here to use the following easy consequence.

\begin{proposition}\label{prop:treecut}
Let $D$ be a complete structured DNNF structured by a v-tree $T$ computing a function $f$ in variables $X$. Let $t$ be a node of $T$ and let $Y:= \var(T_t)$ and $Z=X\setminus \var(T_t)$. Then
 \[\log(\wi(D)) \ge \ccd(f, (Y,Z)).\] 
\end{proposition}

\begin{proof}
From Theorem~\ref{thm:ijcai} and the definition of width, it follows directly that the size of any rectangle cover of $f$ respecting $(Y, Z)$ is upper bounded by the width of $D$.
Taking the logarithm on both sides yields the claim.
\end{proof}

In many cases, instead of considering explicit v-trees, it is more convenient to simply use best-case communication complexity.

\begin{corollary}\label{cor:bestcase}
 Let $f$ be a Boolean function in variables $X$. Then, for every complete structured DNNF computing $f$, we have 
 \[ \wi(D)\ge 2^{\cc(f)}.\]
\end{corollary}
\begin{proof}
 Note that for every v-tree with $X$ on the leaves, there is a node $t$ such that $|X|/3 \le |\var(T_t)| \le 2|X|/3$. Plugging this into Proposition~\ref{prop:treecut} directly yields the result.
\end{proof}

We will use Corollary~\ref{cor:bestcase} to turn compilation algorithms that produce complete structured DNNF based on a parameter of the input as in~\cite{AmarilliCMS18,BovaS17} into inexpressivity bounds based on this parameter. We first give an abstract version of this result that we will instantiate for concrete measures later on.

\begin{theorem}\label{thm:abstract}
 Let $\mathcal{C}$ be a (fully expressive) representation language for Boolean functions. Let $\pp$ be a parameter $\pp:\mathcal{C}\rightarrow \mathbb{N}$. Assume that there is for every Boolean function $f$ and every $C\in \mathcal{C}$ that encodes $f$ a complete structured DNNF with 
 \[\wi(D) \le 2^{\pp(C)}.\]
 Then we have
 \[\pp(C) \ge \cc(f).\]
\end{theorem}
\begin{proof}
 From the assumption, we get $\pp(C)\ge \log(\wi(D))$. Then we apply Corollary~\ref{cor:bestcase} to directly get the result.
\end{proof}

Intuitively, it is exactly the algorithmic usefulness of parameters that makes the resulting instances inexpressive. Note that it is not surprising that instances whose expressiveness is severely restricted allow for good algorithmic properties. However, here we see that the inverse of this statement is also true in a quite harsh way: if a parameter has good algorithmic properties allowing efficient compilation into DNNF, then this parameter puts strong restrictions on the complexity of the expressible functions.

Note that instead of Corollary~\ref{cor:bestcase} we could have used Proposition~\ref{prop:treecut} in the proof of Theorem~\ref{thm:abstract} to get a slightly stronger result. We chose to go with a simpler statement here but note that we will use the extended strength of Proposition~\ref{prop:treecut} later on in Section~\ref{sec:examples}.

From Theorem~\ref{thm:abstract}, we directly get lower bounds for the width measures studied in~\cite{PaulusmaSS13,SamerS10,FischerMR08,SlivovskyS13,SaetherTV15}. The first result considers the parameters with respect to which SAT is fixed-parameter tractable.

\begin{corollary}\label{cor:twlower}
 There is a constant $b>0$ such that for every Boolean function $f$ and every CNF $C$ encoding $f$ we have
 \[\min\{\twi(C),\twp(C), \twd(C), \scw(C)\} \ge b\cdot \cc(f).\]
\end{corollary}
\begin{proof}
 This follows directly from Theorem~\ref{thm:abstract} and the fact that for all these parameters there are algorithms that, given an input CNF of parameter value $k$, construct an equivalent complete structured DNNF of width $2^{O(k)}$.
\end{proof}

Using the compilation algorithm from~\cite{AmarilliCMS18,AmarilliMS18}, we get essentially the same result for circuit representations.
\begin{corollary}\label{cor:twcircuits}
 There is a constant $b>0$ such that for every Boolean function $f$ and every circuit $C$ encoding $f$ we have
 \[\min\{\tw(C),\cw(C)\} \ge b\cdot \cc(f).\]
\end{corollary}
We remark that for treewidth $1$ the circuits of Corollary~\ref{cor:twcircuits} boil down to so-called read-once functions which have been studied extensively, see e.g.~\cite{GolumbicG11}.

Finally, we give a version for parameters that allow polynomial time algorithms when fixed but no fixed-parameter algorithms.

\begin{corollary}\label{cor:cwlower}
 There is a constant $b>0$ such that for every Boolean function $f$ in $n$ variables and every CNF $C$ encoding $f$ we have
 \[\min\{\mimw(C),\cw(C), \mtw(C)\} \ge b\cdot \frac{\cc(f)}{\log(n)}.\]
\end{corollary}
\begin{proof}
 All of the width measures in the statement allow compilation into complete structured DNNF of size -- and thus also width -- $n^{O(k)}$ for parameter value~$k$ and $n$ variables~\cite{BovaCMS15}. Thus, with Theorem~\ref{thm:abstract}, for each measure there is a constant $b'$ with $\log(n^k) = k\log(n) \ge b'\cc(f)$ which completes the proof.
\end{proof}

Note that the bounds of Corollary~\ref{cor:cwlower} are lower by a factor of $\log(n)$ than those of Corollary~\ref{cor:twlower}. We will see in the next section that in a sense this difference is unavoidable.

\section{Relations between Different Width Measures of Encodings}

In this section, we will show that the different width measures for optimal clausal encodings are strongly related.
To this end, in different subsections, we will show the relation of treewidth to all other width measures we consider.
We will then combine these relationships between treewidth and other width measures to analyze the relationships between all width measures we consider.

\subsection{From Treewidth to Modular Treewidth and Cliquewidth}

We will start by proving that primal treewidth bounds imply bounds for modular treewidth and cliquewidth.
% \todo{relations to \cite{Mengel16}}

\begin{theorem}\label{thm:cliquegood}
 Let $k$ be a positive integer and $f$ be a Boolean function of $n$ variables that has a CNF-encoding $F$ of primal treewidth at most $k\log(n)$. Then $f$ also has a CNF-encoding $F'$ of modular incidence treewidth and cliquewidth $O(k)$. Moreover, if $F$ has dependent auxiliary variables, then so has $F'$.
\end{theorem}

Before we prove Theorem~\ref{thm:cliquegood}, let us here discuss this result a little. It is well known that the modular treewidth and the cliquewidth of a CNF formula can be much smaller than its treewidth~\cite{SlivovskyS13}. Theorem~\ref{thm:cliquegood} strengthens this by saying essentially that for \emph{every} function we can gain a factor logarithmic in the number of variables.

In particular, this shows that the lower bounds we can get from Corollary~\ref{cor:cwlower} are the best possible: the maximal lower bounds we can show are of the form $n/\log(n)$ and since there is always an encoding of every function of treewidth~$n$, by Theorem~\ref{thm:cliquegood} there is always an encoding of cliquewidth roughly $n/\log(n)$. Thus the maximal lower bounds of Corollary~\ref{cor:cwlower} are tight up to constants.

Note that for Theorem~\ref{thm:cliquegood}, it is important that we are allowed to change the encoding. For example, the primal graph of the formula $F= \bigwedge_{i,j\in [n]}(x_{i,j}\lor x_{i+1, j}) \land (x_{i,j}\lor x_{i, j+1})$ has the $n\times n$-grid as a minor and thus treewidth $n$, see e.g.~\cite[Chapter 12]{Diestel12}. But the incidence graph of $F$ has no modules and also has the $n\times n$-grid as a minor, so $F$ has modular incidence treewidth at least $n$ as well. So we gain nothing by going from primal treewidth to modular treewidth without changing the encoding. What Theorem~\ref{thm:cliquegood} tells us is that there is a different formula $F'$ that encodes the function of $F$, potentially with some additional variables, such that the treewidth of $F'$ is at most $O(n/ \log(n))$.

Let us note that encodings with dependent auxiliary variables are often useful, e.g. when considering counting problems. In fact, for such clausal encodings, the number of models is the same as for the function they encode. It is thus interesting to see that dependence of the auxiliary variables can be maintained by the construction of Theorem~\ref{thm:cliquegood}. We will see that this is also the case for most other constructions we make.

\begin{proof}[Proof (of Theorem~\ref{thm:cliquegood})]
 The basic idea is that we do not treat the variables in the bags of the tree decomposition individually but organize them in groups of size $\log(n)$. We then simulate the clauses of the original formula by clauses that work on the groups. Since for every group there are only a linear number of assignments, all encoding sizes stay polynomial. We now give the details of the proof.
 
 Let $(T,(B_t)_{t\in T})$ be a tree decomposition of $F$ of width at most $k\log(n)$. For every clause $C$ of $F$ there is a bag $\lambda(C)$ that contains the variables of $C$. By adding some copies of bags, we may assume w.l.o.g.~that for every bag $B$ there is at most one clause with $\lambda(C)= B$ and call this clause $\lambda^{-1}(B)$. 
 
 In a first step, we construct a coloring $\mu: \var(F) \rightarrow [k+1]$ such that in every bag there are at most $\log(n)$ variables of every color. This can be done iteratively as follows: first split the bag $B_r$ at the root $r$ into color classes as required. Since there are at most $k\log(n)+1$ variables in $B_r$ by assumption, we can split them into $k+1$ color classes of size at most $\log(n)$ arbitrarily. Now let $t$ be a node of $T$ with parent $t'$. By the coloring of the variables in $B_{t'}$, some of the variables in $B_t$ are already colored. We simply add the variables not appearing in $B_{t'}$ arbitrarily to color classes such that no color class is too big. Again, since $B_t$ contains at most $k\log(n)+1$ variables, this is always possible. Moreover, due to the connectivity condition, there is for every variable $x$ a unique node $t_x$ that is closest to the root under the bags containing $x$. Consequently, we can make no contradictory decisions during this coloring process, so $\mu$ is well-defined.
 
 We now construct $F'$.
 To this end, we first introduce for every variable $x$ and every node $t$ such that $x\in B_t$ a new variable $x_t$. Now for every node $t$ with parent $t'$ and every color $i$, we add a set $\calC_{t', t, i}$ of clauses in all variables $x_t, x_{t'}$ with $\mu(x)=i$. We construct these clauses in such a way that they are satisfied by exactly the assignments in which for each pair $x_t, x_{t'}$ such that both these variables exist, both variables take the same value. Note that the clauses in $C_{t, t', i}$ have at most $2\log(n)$ variables, so there are at most $n^2$ of them. Moreover, they contain all the same variables. The result is a formula in which all $x_t$ for a variable $x$ take the same value in all satisfying assignments.
 
 In a next step, we do for each clause $C$ the following: let $t = \lambda(C)$. For every color $i$, we define $X_{i, t}$ to be the set of variables $x_t$ such that $\mu(x) = i$. We add a fresh variable $y_{C,i}$ and clauses $\calC_{C,i}$ in the variables $X_{i,t}\cup \{y_{C,i}\}$ that accept exactly the assignments $a$ with 
 \begin{itemize}
  \item $a(y_{C,i})=1$ and there is an $x_t\in X_{i,t}$ such that setting  $x$ to $a(x_t)$ satisfies $C$, or
  \item $a(y_{C,i})=0$ and there is no $x_t\in X_{i,t}$ such that setting  $x$ to $a(x_t)$ satisfies $C$.
 \end{itemize}
Next, we add the clause $C' = \bigvee_{i\in [k+1]} y_{C,i}$. Finally, for every variable $x$, rename one arbitrary variable $x_t$ to $x$. This completes the construction of $F'$.

We claim that $F'$ is an encoding of $f$. To see this, first note that, as discussed before, for every variable $x$ of $F$, in the satisfying assignments of $F'$, all $x_t$ and $x$ take the same value. So, we define for every assignment $a$ of $F$ a partial assignment $a'$ of $F'$ as an extension of $a$ by setting $a'(x_t) = a(x)$ for every $x_t$. $a$ satisfies a clause $C$ if and only if there is at least one variable $x$ of $C$ such that $a(x)$ makes $C$ true. Let $\mu(x)=i$, then $a$ satisfies $C$ if and only if $\calC_{C,i}$ is satisfied by the extension of $a'$ that sets $y_{C,i}$ to $1$. So $a$ satisfies $C$ if and only if there is an extension of $a'$ that satisfies $\calC_{C,i}$. Consequently, $a$ satisfies $F$ if and only if there is an extension $a''$ of $a$ that satisfies $F'$, so $F'$ is an encoding of $f$ as claimed.

To see that the construction maintains dependence of auxiliary variables, observe first that the auxiliary variables already present in $F$ are still in~$F'$ and they are still dependent. We claim that all the new variables depend on those of~$F$. For the variables $x_t$, this is immediate since they must take the same value as $x$ in every model. Moreover, the variables $y_{C,i}$ depend on the $x_t$ by definition. As a consequence, all auxiliary variables are dependent

We now show that the modular treewidth of $F'$ is at most $O(k)$. First note that all sets $X_{i,t}$ are modules as are the clause sets $\calC_{t,t',i}$ and $\calC_{C,i}$. W.l.o.g.~we may assume that for every $t$, there is at most one clause $C$ with $\lambda(C)=t$ and that $T$ is a binary tree. We construct a tree decomposition $(T, (B_t)_{t\in V(T)})$ as follows: we put a representant of $X_{i,t}$, $\calC_{t,t', i}$, $\calC_{t', t,i}$ and $\calC_{C,i}$ into $B'_t$. Moreover, we add $y_{C,i}$ and $C'$ to $B'_t$. It is easy to see that constructed like this, $(T, (B_t)_{t\in V(T)})$ is a tree decomposition of width at most $O(k)$.

Finally, we will show that the incidence graph of the formula $i$ can be constructed with $O(k)$ labels. In this construction, the relabeling operation will only ever be used to \emph{forget} labels, i.e., we change a label $i$ into a global dummy label $d$ such that vertices labeled by $d$ are never used in joining operations. 

In a first step, we color $T$ with $4$ colors such that for every node $t$, the node $t$, its at most two children and its parent all have different colors. We denote the color of $t$ by $\eta(t)$. Then, for every $t$ individually, we create the nodes in $\calC_{C,i}$, $X_{t,i}$ where $C$ is such that $\lambda(C) = t$. The clauses in $\calC_{C,I}$ get label $(i,\eta(t),0)$ and the variables in $X_{t,i}$ get label $(i, \eta(t), 1)$. By joining the vertices with labels $(i, \eta(t), 0)$ with those with $(i, \eta(t), 1)$, we connect the variables in $X_{t,i}$ with the clauses in $\calC_{C, i}$. We then create the $y_{C,i}$, each with individual labels and connect them to the clauses with label $(i, \eta, 0)$. Finally, we create the clause vertex $C'$ with an individual label and connect it to the $y_{C,i}$. We then forget the labels of all vertices except the $X_{t,i}$. We call the resulting graph $G_t$.

Note that at this point, the only thing that remains to do is to introduce the clauses in the $\calC_{t, t', i}$ and connect them to the variables in $G_t$ and $G_{t'}$. To do so, we work in a bottom-up fashion along $T$. For the leaves of $T$, there is nothing to do. So let $t$ be an internal node of $T$ with children $t_1, t_2$; the case in which $t$ only has one child is treated analogously. By induction, we assume that we have graphs $G'_{t_1}$ and $G'_{t_2}$ containing $G_{t_1}$ and $G_{t_2}$ as respective subgraphs such that:
\begin{itemize}
 \item all variables appearing in $G'_{t_j}$ are already connected to all clauses, except the variables in the $X_{t_j, i}$ which are not yet connected to the clauses $\calC_{t, t_1, i}$,
 \item all vertices in $G'_{t_j}$ except for those in the $X_{t_i}$ have the dummy label $d$.
\end{itemize}
We proceed as follows: we make a disjoint union of $G_t$, $G'_{t_1}$ and $G'_{t_2}$. Then we create nodes for all clauses in the $\calC_{t, t_1, i}$ giving them the label $(i, \eta(t), 2)$. Then we connect all nodes with label $(i, \eta(t_1), 1)$ to those with label $(i, \eta(t), 2)$, i.e., we connect the nodes in $X_{t_1, i}$ with the clauses in $\calC_{t, t_1, i}$. Then we connect all nodes with label $(i, \eta(t), 1)$ to those with label $(i, \eta(t), 2)$, i.e., we connect the nodes in $X_{t, i}$ with the clauses in $\calC_{t, t_1, i}$. We proceed analogously with $t_2$. Finally, we forget all labels but those for the $X_{t, i}$. This completes the construction. 

Verifying the clauses in $F'$, one can see that the resulting graph is indeed the incidence graph of $F'$. Moreover, we have only used $O(k)$ clauses by construction. This completes the proof.
\end{proof}

\subsection{Back to Treewidth}

We now show that the reverse of Theorem~\ref{thm:cliquegood} is also true: upper bounds for many width measures imply also bounds for the primal treewidth of clausal encodings. Note that this is at first sight surprising since without auxiliary variables many of those width measures are known to be far stronger than primal treewidth.

\begin{theorem}\label{thm:reverse}
 Let $f$ be a Boolean function of $n$ variables. 
 \begin{enumerate} 
  \item[a)] If $F$ has a clausal encoding of modular treewidth, cliquewidth or mim-width $k$ then $f$ also has a clausal encoding $F'$ of primal treewidth $O(k\log(n))$ with $O(k n\log(n))$ auxiliary variables and $n^{O(k)}$ clauses.
  \item[b)] If $F$ has a clausal encoding of incidence treewidth, dual treewidth, or signed incidence cliquewidth $k$, then $f$ also has a clausal encoding $F'$ of primal treewidth $O(k)$ with $O(nk)$ auxiliary variables and $2^{O(k)} n$ clauses.
 \end{enumerate}
\end{theorem}

To show Theorem~\ref{thm:reverse} and several similar results for other width measures in this section, we make a detour through DNNF.
The idea is to show that from certain DNNF representations of functions, we can get clausal encondings of primal treewidth strongly related to the width of the DNNF.
Since many width measures can be used to construct small width DNNFs, we get small width clausal encodings for these width measures.
We now give a precise statement of the relation between DNNF and treewidth of clausal encodings.

\begin{lemma}\label{lem:dnnftotw}
 Let $f$ be a Boolean function in $n$ variables that is computed by a complete structured DNNF of width $k$. Then $f$ has a clausal encoding $F$ of primal treewidth $9\log(k)$ with $O(n \log(k))$ variables and $O(n k^3)$ clauses. Moreover, if $D$ is deterministic then $F$ has dependent auxiliary variables.
\end{lemma}

The proof of Lemma~\ref{lem:dnnftotw} will rely on so-called proof trees in DNNF, a concept that has found wide application in circuit complexity and in particular also in knowledge compilation. To this end, we make the following definition:
a \emph{proof tree} $\calT$ of a complete structured DNNF $D$ is a circuit constructed as follows:
\begin{enumerate}
 \item The output gate of $D$ belongs to $\calT$.
 \item \label{enu:or} Whenever $\calT$ contains an $\lor$-gate, we add exactly one of its inputs.
 \item Whenever $\calT$ contains an $\land$-gate, we add both of its inputs.
 \item No other gates are added to $\calT$.
\end{enumerate}
Note that the choice in Step~\ref{enu:or} is non-deterministic, so there are in general many proof trees for $D$. Observe also that due to the structure of $D$ given by its v-tree, every proof tree is in fact a tree which justifies the name. Moreover, letting $T$ be the v-tree of $D$, every proof tree of $D$ has exactly one $\lor$-gate and one $\land$-gate in the set $\mu(t)$ for every non-leaf node $t$ of $T$. For every leaf $t$, every proof tree contains an input gate $x$ or $\neg x$ where $x$ is the label of $t$ in $T$.

The following simple observation that can easily be shown by using distributivity is the main reason for the usefulness of proof trees.

\begin{observation}
 Let $D$ be a complete structured DNNF and $a$ an assignment to its variables. Then $a$ satisfies $D$ if and only if it satisfies one of its proof trees. Moreover, if $D$ is deterministic, then every assignment $a$ that satisfies $D$ satisfies exactly one proof tree of $D$.
\end{observation}

\begin{proof}[Proof (of Lemma~\ref{lem:dnnftotw})]
Let $D$ be the complete structured DNNF computing $f$ and let $T$ be the v-tree of $D$.
 The idea of the proof is to use auxiliary variables to ``guess'' for every $t$ an $\lor$-gate and an $\land$-gate. Then we use clauses along the v-tree~$T$ to verify that the guessed gates in fact form a proof tree and check in the leaves of $T$ if the assignment to the variables of $f$ satisfies the encoded proof tree. We now give the details of the construction.
 
 We first note that, as shown in~\cite{CapelliM18}, in complete structured DNNF of width $k$, one may assume that every set $\mu(t)$ contains at most $k^2$ $\land$-gates so we assume this to be the case for $D$. For every node $t$ of $T$, we introduce a set $X_t$ of $3\log(k)$ auxiliary variables to encode one $\lor$-gate and one $\land$-gate of $\mu(t)$ if $t$ is an internal node. If $t$ is a leaf, $X_t$ encodes one of the at most $2$ input gates in~$\mu(t)$. We now add clauses that verify that the gates chosen by the variables $X_t$ encode a proof tree by doing the following for every $t$ that is not a leaf: first, add clauses in $X_t$ that check if the chosen $\land$-gate is in fact an input of the chosen $\lor$-gate. Since $X_t$ has at most $3\log(k)$ variables, this introduces at most $k^3$ clauses. Let $t_1$ and $t_2$ be the children of $t$ in $T$. Then we add clauses that verify if the $\land$-gate chosen in $t$ has as input either the $\lor$-gate chosen in $t_1$ if $t_1$ is not a leaf, or the input gate chosen in $t_1$ if $t_1$ is a leaf. Finally, we add analogous clauses for $t_2$. Each of these clause sets is again in $3\log(k)$ variables, so there are at most $2k^3$ clauses in them overall. The result is a CNF-formula that accepts an assignment if and only if it encodes a proof tree of $D$.
 
 We now show how to verify if the chosen proof tree is satisfied by an assignment to $f$. To this end, for every leaf $t$ of $T$ labeled by a variable $x$, add clauses that check if an assignment to $x$ satisfies the corresponding input gate of $D$. Since $\mu(t)$ contains at most $2$ gates, this only requires at most $4$ clauses. This completes the construction of the clausal encoding. Overall, since $T$ has $n$ internal nodes, the CNF has $n (3\log(k) + 1)$ variables and $3n k^3 + 4n$ clauses.
 
 It remains to show the bound on the primal treewidth. To this end, we construct a tree decomposition $(T,(B_t)_{t\in V(T)})$ with the v-tree $T$ as underlying tree as follows: for every internal node $t\in V(T)$, we set $B_t := X_t \cup X_{t_1} \cup X_{t_2}$ where $t_1$ and $t_2$ are the children of~$t$. Note that for every clause that is used for checking if the chosen nodes form a proof tree, the variables are thus in a bag~$B_t$. For every leaf~$t$, set $B_t := X_t\cup \{x\}$ where $x$ is the variable that is the label of~$t$. This covers the remaining clauses. It follows that all edges of the primal graph are covered. To check the third condition of the definition of a tree decomposition, note that every auxiliary variable in a set $X_t$ appears only in $B_t$ and potentially in $B_{t'}$ where $t'$ is the parent of $t$ in $T$. Thus $(T,(B_t)_{t\in V(T)})$ constructed in this way is a tree decomposition of the primal graph of $C$. Obviously, the width is bounded by $9\log(k)$ since every $X_t$ has size $3\log(k)$, which completes the proof.
\end{proof}

\begin{proof}[Proof (of Theorem~\ref{thm:reverse})]
 We first show a). By~\cite{BovaCMS15}, whenever the function $f$ has a clausal encoding $F$ with one of the width measures from this statement bounded by $k$, then there is also a complete structured DNNF $D$ of width $n^{O(k)}$ computing~$F$. Now forget all auxiliary variables of $F$ to get a DNNF representation $D'$ of~$f$. Note that since forgetting does not increase the width, see~\cite{CapelliM18}, $D'$ also has width at most $n^{O(k)}$. We then simply apply Lemma~\ref{lem:dnnftotw} to get the result.
 
 To see b), just observe that, following the same construction, the width of $D$ is $2^{O(k)}$ for all considered width measures~\cite{BovaCMS15}.
\end{proof}

Remark that the construction of Theorem~\ref{thm:reverse} has a surprising property: the size and the number of auxiliary variables of the constructed encoding $F'$ does \emph{not} depend on the size of the initial encoding at all. Both depend only on the number of variables in $f$ and the width.

To maintain dependence of the auxiliary variables in the above construction, we have to work some more than for Theorem~\ref{thm:reverse}.
We start with some definitions.
We call a complete structured DNNF \emph{reduced} if from every gate there is a directed path to the output gate. Note that every complete structured DNNF can be turned into a reduced DNNF in linear time by a simple graph traversal and that this transformation maintains determinism and structuredness by the same v-tree. The following property will be useful.

\begin{lemma}\label{lem:extend}
 Let $D$ be a reduced complete structured DNNF and let $g$ be a gate in $D$. Let $a_g$ be an assignment to $\var(g)$, the variables in the subcircuit rooted in $g$, that satisfies $g$. Then, $a_g$ can be extended to an assignment $a$ that satisfies $D$.
\end{lemma}
\begin{proof}
 We use the fact that an assignment to $D$ is satisfying if and only if there is a proof-tree that witnesses this. So let $\calT_g$ be a proof tree that witnesses $a_g$ satisfying $g$. We extend it to a proof tree for an extension $a$ of $a_g$ as follows: first add a path from $g$ to the output gate to $\calT_g$ and then iteratively add more gates as required by the definition of proof trees where the choices in $\lor$-gates are performed arbitrarily. The result is an extension $\calT$ of $\calT_g$ which witnesses that an assignment $a$ that extends $a_g$ satisfies $D$.
\end{proof}

Let $f$ be a function in variables $X\cup \{z\}$. We say that $z$ is \emph{definable} in $X$ with respect to~$f$ if there is a function $g$ such that for all assignments $a$ with $f(a)=1$ we have $a(z) = g(a|_{X})$ where $a|_{X}$ is the restriction of $a$ to $X$. 
\begin{lemma}\label{lem:forget}
Let $f$ be a function in variables $X\cup \{z\}$ such that $z$ is definable in $X$ with respect to~$f$. Let $D$ be a reduced complete structured deterministic $DNNF$ computing $f$. Then the complete structured DNNF $D'$ we get from $D$ by forgetting $z$ is deterministic as well.
\end{lemma}
\begin{proof}
 By way of contradiction, assume this were not the case. Then there is an $\lor$-gate $g$ in $D'$ and an assignment $a'$ to $X$ such that two children $g_1$ and $g_2$ are satisfied by $a'$. By Lemma~\ref{lem:extend}, we may assume that $a'$ satisfies $D'$. Then there are extensions $a_1$ and $a_2$ of $a$ that assign a value to $z$ such that $a_1$ satisfies $g_1$ and $a_2$ satisfies $g_2$ in $D$. Note that both $a_1$ and $a_2$ satisfy $D$ and thus, by definability, $a_1$ and $a_2$ assign the same value to $z$. So $a_1=a_2$ and hence $a_1$ satisfies both $g_1$ and $g_2$ in $D$ which contradicts the determinism of $D$.
\end{proof}

\begin{theorem}\label{thm:reversedet}
 Let $f$ be a Boolean function of $n$ variables. 
 \begin{enumerate} 
  \item[a)] If $F$ has a clausal encoding with dependent auxiliary variables of modular treewidth, cliquewidth or mim-width $k$ then $f$ also has a clausal encoding $F'$ with dependent auxiliary variables of primal treewidth $O(k\log(n))$ with $O(k n\log(n))$ auxiliary variables and $n^{O(k)}$ clauses.
  \item[b)] If $F$ has a clausal encoding with dependent auxiliary variables of incidence treewidth, dual treewidth, or signed incidence cliquewidth $k$, then $f$ also has a clausal encoding $F'$ with dependent auxiliary variables of primal treewidth $O(k)$ with $O(nk)$ auxiliary variables and $2^k n$ clauses.
 \end{enumerate}
\end{theorem}
\begin{proof}
 The proof is essentially the same as that of Theorem~\ref{thm:reverse} with some additional twists. First observe that the complete structured DNNF $D$ constructed with~\cite{BovaCMS15} is deterministic. Then we use Lemma~\ref{lem:forget} when forgetting the auxiliary variables and get a $D'$ that is deterministic without increasing the width. Then, since $D'$ is deterministic, we can construct a clausal encoding with dependent auxiliary variables using Lemma~\ref{lem:dnnftotw}.
\end{proof}

Next we will show that signed incidence cliquewidth is linearly related to primal treewidth when allowing auxiliary variables. We will state a result similar to Lemma~\ref{lem:dnnftotw}.

To do so, we will start with a special case for which we introduce some more definitions: a \emph{special tree decomposition} of a graph $G$ is defined as a tree decomposition $(T, (B_t)_{t\in V(T)})$ in which for every vertex $x\in V(G)$ the set $\{t\in V(T) \mid x\in B_t\}$ lies on a leaf-root path in $T$~\cite{Courcelle12}. The \emph{special treewidth} is defined as the smallest width of any special tree decomposition of $G$. Finally, we define the primal special treewidth of a CNF-formula as the special treewidth of its primal graph.

\begin{lemma}\label{lem:special}
Every CNF-formula of primal special treewidth $k$ has signed incidence cliquewidth at most $k+1$.
\end{lemma}
\begin{proof}
 Let $(T, (B_t)_{t\in V(T)})$ be a special tree decomposition of the primal graph of $F$. It is well known that for every clause $C$ there is a node $t=\lambda(C)$ of $T$ such that all variables of $C$ are in $B_t$. By adding copies of some bags $B_t$ along a root-leaf path in $T$, we may assume that $\lambda(C)\ne \lambda(C')$ for every pair $C, C'$ of clauses with $C\ne C'$.
 
 We will show how to construct the signed incidence graph $G'$ of $F$ with the operations in the definition of signed cliquewidth along the tree $T$. In a first step, we label every variable $x$ of $F$ with a color $\mu(x)$ from $\{1, \ldots, k+1\}$ such that in every bag $B_t$ there are no two variables with the same label $\mu(x)$. This can be done similarly to the first step of the proof of Theorem~\ref{thm:cliquegood} by descending from the root to the leaves and labeling the variables in the bags along this way. The label $\mu(x)$ will be the label that the variable gets when it is created in the construction of $G'$. As in the proof of Theorem~\ref{thm:cliquegood}, the only renamings of labels that we will perform will be forget operations, i.e., renaming a label to a dummy label $d$.
 
 For the construction of $G'$, we will iteratively construct for every $t\in V(T)$ a graph $G_t$ that contains all variables in $S_t:=\bigcup_{t'\in V(T_t)} B_t$ where $T_t$ is the subtree of $T$ rooted in $t$. Moreover, $G_t$ contains all clauses such that $\lambda(C)$ lies in $T_t$ and all signed edges connecting them to their variables.
 
 If $t$ is a leaf, then we create all variables in $B_t$ and if there is a clause $C$ with $\lambda(C)=t$, we introduce it with color $k+2$. Since all variables of $C$ have different colors, we can then introduce all signed edges individually. This completes the construction for the leaf case.
 
 Let now $t$ be an internal node with children $t_1, \ldots , t_\ell$. By assumption, we have already constructed $G_{t_1}, \ldots, G_{t_\ell}$. Note that for every $i$ the variables in $G_{t_i}$ that are not in $B_t$ are by construction already connected to all their clauses in $G_{t_i}$, so we can safely forget their label in a first step. Now we take the disjoint union of all $G_{t_i}$. Note that this union is in fact disjoint, because, since we start from a special tree decomposition, no node appears in more than one $G_{t_i}$. Now we create the variables which appear in $B_t$ but not in any $G_{t_i}$. Note that at this point the vertices with non-dummy labels are exactly those in $B_t$. If there is no clause $C$ with $\lambda(C)=t$, we are done. Otherwise, we create $C$ and connect it to all its variables by signed edges as in the leaf case. This completes the construction of $G_t$.
 
 For the root $r$ of $T$ we have $G_r = G'$ by definition. Moreover, we have used at most $k+2$ labels. This completes the proof.
\end{proof}

With Lemma~\ref{lem:special}, we can give a version of Lemma~\ref{lem:dnnftotw} for signed incidence cliquewidth easily.

\begin{lemma}\label{lem:dnnftoscw}
 Let $f$ be a Boolean function in $n$ variables that is computed by a structured DNNF of width $k$. Then $f$ has a clausal encoding $F$ of signed incidence cliquewidth and primal special treewidth $O(\log(k))$ with $O(n \log(k))$ variables and $O(n k^3)$ clauses. Moreover, if $D$ is deterministic then $F$ has dependent auxiliary variables.
\end{lemma}
\begin{proof}
 We only have to observe that in fact the tree decomposition in the proof of Lemma~\ref{lem:dnnftotw} is special and apply Lemma~\ref{lem:special}.
\end{proof}

\begin{corollary}\label{cor:scw}
 Let $f$ be a function with a CNF-representation of primal tree\-width~$k$. Then $f$ has a clausal encoding of signed incidence cliquewidth and special treewidth~$O(k)$.
\end{corollary}

\subsection{Putting Things Together}

We can now state the main result of this section.

\begin{theorem}
\label{thm:main}
 Let $A=\{\twp, \twd, \twi, \scw\}$ and $B= \{\mtw, \cw, \mimw\}$. Let $f$ be a Boolean function in $n$ variables.
 \begin{itemize}
  \item[a)] Let $w_1\in A$ and $w_2\in B$. Then there are constants $c_1$ and $c_2$ such that the following holds: let $F_1$ and $F_2$ be clausal representations for $f$ with minimal $w_1$-width and $w_2$-width, respectively. Then \[ w_1(F_1) \le k\log(n) \Rightarrow w_2(F_2) \le c_1 k \] and \[w_2(F_2) \le k \Rightarrow w_1(F_1) \le c_2 k\log(n).\]
  \item[b)] Let $w_1\in A$ and $w_2\in A$ or $w_1 \in B$ and $w_2\in B$. Then there are constants $c_1$ and $c_2$ such that the following holds: let $F_1$ and $F_2$ be clausal representations for $f$ of minimal $w_1$-width and $w_2$-width, respectively. Then \[ w_1(F_1) \le k \Rightarrow w_2(F_2) \le c_1 k \] and \[w_2(F_2) \le k \Rightarrow w_1(F_1) \le c_2 k.\]
 \end{itemize}
\end{theorem}
\begin{proof}
 Assume first that $w_1=\twp$. For a) we get the second statement directly from Theorem~\ref{thm:reverse} a). For $\cw$ and $\mtw$ we get the first statement by Theorem~\ref{thm:cliquegood}. For $\mimw$ it follows by the fact that for every graph $\mimw(G) \le c \cdot \cw(G)$ for some absolute constant~$c$, see~\cite[Section 4]{vatshelleThesis}.
 
 For b), the second statement is Theorem~\ref{thm:reverse} b). Since for every formula $F$ we have $\twi(F) \le \twp(F)+1$, see e.g.~\cite{FischerMR08}, the first statement for $\twi$ is immediate. For $\scw$ it is shown in Corollary~\ref{cor:scw}, while for $\twd$ it can be found in~\cite{SamerS10b}.
 
 All other combinations of $w_1$ and $w_2$ can now be shown by an intermediate step using $\twp$.
\end{proof}

\section{Applications}\label{sec:examples}

\subsection{Cardinality Constraints}

In this section, we consider cardinality constraints, i.e., constraints of the form $\sum_{i\in [n]} x_i \le k$ in the Boolean variables $x_1, \ldots, x_n$. The value $k$ is commonly called the \emph{degree} or the \emph{threshold} of the constraint. Let us denote by $C_n^k$ the cardinality constraint with $n$ variables and degree $k$. Cardinality constraints have been studied extensively and many encodings are known, see e.g.~\cite{Sinz05}. Here we add another perspective on cardinality constraint encodings by determining their optimal treewidth. We remark that we could have studied cardinality constraints in which the relation is $\ge$ instead of $\le$ with essentially the same results.

We start with an easy observation:

\begin{observation}\label{obs:cardinality}
 $C_n^k$ has an encoding of primal treewidth $O(\log(\min(k, n-k)))$
\end{observation}
\begin{proof}
First assume that $k < n/2$. We iteratively compute the partial sums of $S_j:=\sum_{i\in [j]} x_i$ and encode their values in $\log(k)+1$ bits $Y^j:=\{y_1^j, \ldots , y_{\log(k)+1}^j\}$. We cut these sums off at $k+1$ (if we have seen at least $k+1$ variables set to $1$, this is sufficient to compute the output). In the end we encode a comparator comparing the last sum $S_n$ to $k$.

Since the computation of $S_{j+1}$ can be done from $S_j$ and $x_{j+1}$, we can compute the partial sums with clauses containing only the variables in $Y^j \cup Y^{j+1} \cup \{x_{j+1}\}$, so $O(\log(k))$ variables. The resulting CNF-formula can easily be seen to be of treewidth $O(\log(k))$.

If $k > n/2$, we proceed similarly but count variables assigned to $0$ instead of those set to $1$.
\end{proof}

We remark that our construction is that described as the basic approach in~\cite[Section 8.6.7]{Handbook09}. It has some similarity with the sequential counter introduced in~\cite{Sinz:2005:TOC:3102585.3102662}.
The main difference is that we encode the partial sums $S_j$ in binary whereas in the sequential counter, they are encoded in unary.
This latter encoding has better properties with respect to unit propagation, whereas our encoding has smaller treewidth, which is the parameter we are optimizing for.
We now show that Observation~\ref{obs:cardinality} is essentially optimal. 

\begin{proposition}Let $k< n/2$. Then
 \[\cc(C_n^k) = \Omega(\log(\min(k, n/3))).\]
\end{proposition}

\begin{proof}
Let $s = \min(k, \frac{n}{3})$.
 Consider an arbitrary partition $Y,Z$ with $\frac{n}{3} \le |Y| \le \frac{2n}{3}$. We show that every rectangle cover of $C_n^k$ must have $s$ rectangles. To this end, choose assignments $(a_0, b_0), \ldots, (a_s, b_s)$ such that $a_i:Y\rightarrow \{0,1\}$ assigns $i$ variables to $1$ and $b_i:Z\rightarrow \{0,1\}$ assigns $k-i$ variables to $1$. Note that every $(a_i, b_i)$ satisfies $C_n^k$. We claim that no rectangle $r_1(Y) \land r_2(Z)$ in a rectangle cover of $C_n^k$ can have models $(a_i, b_i)$ and $(a_j, b_j)$ for $i \ne j$. To see this, assume that such a model exists and that $i<j$. Then the assignment $(a_j, b_i)$ is also a model of the rectangle since $a_j$ satisfies $r_1(Y)$ and $b_i$ satisfies $r_2(Z)$. But $(a_j, b_i)$ contains more than $k$ variables assigned to $1$, so the rectangle $r_1(Y) \land r_2(Z)$ cannot appear in a rectangle cover of $C_n^k$. Thus, every rectangle cover of $C_n^k$ must have a different rectangle for every model $(a_i, b_i)$ and thus at least $s$ rectangles. This completes the proof for this case.
\end{proof}

A symmetric argument shows that for $k> n/2$ we have the lower bound $\cc(C_n^k) = \Omega(\log(\min(n-k, n/3)))$.
Observing that $k<n$ for non-trivial cardinality constraints, we get the following from Theorem~\ref{thm:ijcai}.

\begin{corollary}
Clausal encodings of smallest primal treewidth for $C_n^k$ have primal treewidth $\Theta(\log(s))$ for $s=\min(k, n-k)$. The same statement is true for dual and incidence treewidth and signed incidence cliquewidth.
For incidence cliquewidth, modular treewidth and mim-width, there are clausal encodings of $C_n^k$ of constant width.
\end{corollary}

\subsection{The Permutation Function}

We now consider the permutation function~\perm{} which has the $n^2$ input variables $X_n=\{x_{ij} \mid i,j\in [n]\}$ thought of as a matrix in these variables. \perm{} evaluates to $1$ on an input $a$ if and only if $a$ is a permutation matrix, i.e., in every row and in every column of $a$ there is exactly one $1$. 

\begin{example}
 The function \constraint{PERM}$_2$ has the variables $x_{11}, x_{12}, x_{21}, x_{22}$ which we interpret organized as the matrix 
 $\begin{pmatrix}
   x_{11}& x_{12}\\ x_{21}& x_{22}
  \end{pmatrix}$. 
The only inputs on which \constraint{PERM}$_2$ evaluates to $1$ are 
 $\begin{pmatrix}
   1& 0\\ 0& 1
  \end{pmatrix}$ and
 $\begin{pmatrix}
   0& 1\\ 1& 0
  \end{pmatrix}$. 
Inputs on which \constraint{PERM}$_2$ evaluates to $0$ are for example 
 $\begin{pmatrix}
   1& 1\\ 1& 0
  \end{pmatrix}$ (the first row has more than one $1$-entry) and 
 $\begin{pmatrix}
   0& 1\\ 0& 0
  \end{pmatrix}$ 
(the first column has no $1$-entry).
\end{example}

\perm{} is known to be hard in several versions of branching programs, see~\cite{Wegener00}. In~\cite{BriquelKM11}, it was shown that clausal encodings of \perm{} require tree\-width~$\Omega(n/\log(n))$. We here give an improvement by a logarithmic factor.

\begin{lemma}\label{lem:ccperm}
 For every v-tree $T$ on variables $X_n$, there is a node $t$ of $T$ such that \[\ccd(\text{\perm}, Y, Z) = \Omega(n)\] where $Y=\var(T_t)$ and $Z= X\setminus Y$.
\end{lemma}

\begin{proof}
The proof is a variation of arguments in~\cite{BriquelKM11} and in~\cite{Krause88}, see also~\cite[Section 4.12]{Wegener00}. Since all models of \perm{} assign exactly $n$ variables to $1$, for every model $a$ of \perm{} there is a node $t_a$ in $T$ such that $T_t$ contains between $n/3$ and $2n/3$ variables assigned to $1$ by $a$. Since $T$ has $n$ internal nodes and \perm{} has $n!$ models, there must be a node $t$ such that for at least $(n-1)!$ of the models $a$ we have $t = t_a$. We will show in the remainder that $t$ has the desired property.

Denote by $A$ the set of models of \perm{} for which $t_a = t$. Let $Y= \var(T_t)$ and $Z= X_n\setminus Y$ as in the statement of the lemma. Every model $a$ of \perm{} corresponds to a permutation $\pi_a$ on $[n]$ that assigns every $i\in [n]$ to the $j$ such that $a(x_{ij})=1$. Note that because of the properties of $a$, $\pi_a$ is well-defined and indeed a permutation.

Let $R(X) = r_1(Y)\land r_2(Z)$ be a rectangle in a rectangle cover of \perm{} with partition $(Y,Z)$. We will show that $R(X)$ contains few models from $A$. To this end, fix a model $a\in A$ of $R(X)$ and define $I(a) = \{i \mid x_{i, \pi_a(i)} \in Y\}$. Note that $k:= |I(a)|$ is the number of variables in $Y$ that are assigned to $1$ by $a$ and thus $n/3 \le k \le 2n/3$. Let $a'$ be another model of $R(X)$. Then $I(a') = I(a)$ because otherwise $a|_Y \cup a'|_Z$ does not encode a permutation where $a|_Y$ denotes the restriction of $a$ to $Y$ and $a'|_Z$ that of $a'$ to $Z$. Letting $I'(a)= \{\pi_a(i)\mid i\in I(a)\}$, we get similarly that for all models $a'$ of $R(X)$ we have $I'(a)=I'(a')$. It follows that the models of $r_1(Y)$ are all bijections between $I(a)$ and $I'(a)$ and thus $r_1(Y)$ has at most $k!$ models.

By a symmetric argument, one sees that $r_2(Z)$ has at most $(n-k)!$ models. Thus, the number of models of $R$ is bounded by $k!(n-k)! \le \left(\frac n 3\right)! \left(\frac{2n} 3\right)!$. As a consequence, to cover all $(n-1)!$ models in $a$, one needs at least 
\[\frac{(n-1)!}{\left(\frac{n}{3}\right)! \left(\frac{2n} 3\right)!} = \frac{1}{n} \binom{n}{\frac{n}{3}} \ge \frac 1 n \left(\frac{n}{\frac{n}{3}}\right)^{\frac{n}{3}} = \frac 1 n \sqrt[3]{3} ^n
\]
rectangles, which completes the proof.
\end{proof}

As a consequence of Lemma~\ref{lem:ccperm}, we get an asymptotically tight treewidth bound for encodings of \perm.
\begin{corollary}\label{cor:permtw}
 Clausal encodings of smallest primal treewidth for $C_n^k$ have primal treewidth $\Theta(n)$. 
\end{corollary}
\begin{proof}[Proof (sketch)]
 The lower bound follows by using Lemma~\ref{lem:ccperm} and Proposition~\ref{prop:treecut} and then arguing as in the proof of Theorem~\ref{thm:abstract}.
 
 For the upper bound, observe that checking if out of $n$ variables exactly one has the value $1$ can easily be done with $n$ variables. We apply this for every row in a bag of a tree decomposition. We perform these checks for one row after the other and additionally use variables for the columns that remember if in a column we have seen a variable assigned $1$ so far. Overall, to implement this, one needs $O(n^2)$ auxiliary variables and gets a formula of treewidth $O(n)$.
\end{proof}

From Corollary~\ref{cor:permtw} we get the following bound by applying Theorem~\ref{thm:cliquegood}. This answers an open problem from~\cite{BriquelKM11} which showed only conditional lower bounds for the incidence cliquewidth of encodings of \perm.

\begin{corollary}
  Clausal encodings of smallest incidence cliquewidth for $C_n^k$ have width $\Theta(n/log(n))$. 
\end{corollary}

\subsection{Improved Lower Bounds for Minor-Free Graphs}

In this section, we show how our approach can be used to improve lower bounds for structurally restricted classes of circuits. We recall that a \emph{minor} $H$ of a graph $G$ is a graph that we can get from $G$ by deleting vertices, deleting edges and contracting edges. For a graph $H$, the class of $H$-minor-free graphs is defined as the class of graphs consisting of all graphs that do not have $H$ as a minor. $H$-minor-free graphs have been studied extensively. In particular, it is known that for planar graphs, and more generally for all graphs embeddable in a fixed surface, there is a graph $H$ such that those graphs are $H$-minor free. For example, planar graphs are $K_5$-minor-free and $K_{3,3}$-minor-free.

We say that a Boolean circuit $C$ is $H$-minor-free if the underlying undirected graph of $C$ is $H$-minor-free. Remember that we assume that every input variable is the label of at most one input gate. There have long been quadratic lower bounds for planar circuits~\cite{LiptonT80}. Those were generalized to almost quadratic lower bounds of the order~$\Omega(n^2/\log(n)^2)$ for $H$-minor-free circuits in~\cite{Oliveira18}. We show here that with our techniques it is easy to improve these bounds to quadratic lower bounds.

As in~\cite{Oliveira18}, the basic building block for our lower bound will be the following result on the treewidth of $H$-minor-free graphs.

\begin{theorem}[\cite{AlonST90}]\label{thm:alonetal}
 For every graph $H$ there is a constant $h$ such that every $H$-minor-free graph $G$ has treewidth at most $h \sqrt{|V(G)|}$.
\end{theorem}

\begin{corollary}
 For every graph $H$ there is a constant $h'$ such that for every function $f$, every $H$-minor-free circuit $C$ computing $f$ has at least $\cc(f)^2$ gates.
\end{corollary}
\begin{proof}
 By Corollary~\ref{cor:twcircuits}, any circuit computing $f$ must have treewidth~$\Omega(\cc(f))$. By Theorem~\ref{thm:alonetal}, the treewidth of $C$ is at most $\sqrt{s}$ where $s$ is the number of gates in $C$. Thus $\sqrt{s} \ge \cc(f)$ and the claim follows.
\end{proof}

To show a quadratic lower bound, consider the function \constraint{$\triangle$-free$_n$} in variables $X_{ij}$ with $1 \le i < j \le n$ which is defined as follows: interpret the input as the adjacency matrix of a graph $G$ and return $1$ if and only if $G$ does not have a triangle as a subgraph.  We note that \constraint{$\triangle$-free$_n$} is a classical function, considered in communication complexity essentially since the creation of the field~\cite{PapadimitriouS84}. Here, we will use the following result:

\begin{theorem}[\cite{JuknaS02}]\label{thm:triangles}
The best-case non-deterministic communication complexity with $\frac 1 3$-balance of \constraint{$\triangle$-free$_n$} is quadratic in $n$, i.e.,
\[\cc(\textrm{\constraint{$\triangle$-free$_n$}}) = \Omega(n^2).\]
\end{theorem}

We directly get the following generalization of the quadratic lower bound in~\cite{LiptonT80}, which improves that in~\cite{Oliveira18}.

\begin{theorem}
 For every fixed graph $H$ there is a constant $h'$ such that every $H$-minor-free circuit computing \constraint{$\triangle$-free$_n$} has $\Omega(n^4)$ gates, i.e., quadratic in the number of inputs.
\end{theorem}

\section{Conclusion}

We have shown several results on the expressivity of clausal encodings with restricted underlying graphs. In particular, we have seen that many graph width measures from the literature put strong restrictions on the expressivity of encodings. We have also seen that, contrary to the case of representations by CNF-formulas, in the case where auxiliary variables are allowed, all width measures we have considered are strongly related to primal treewidth and never differ by more than a logarithmic factor. Moreover, most of our results are also true while maintaining dependence of auxiliary variables.

From a practical standpoint, one point of our results might be that formulas solved with width-based algorithms as those from the theoretical literature can likely only deal with quite simple formulas. Otherwise, for example if formulas contain big cardinality constraints or pseudo-Boolean constraints, the width of the formulas might be infeasibly high. This is because all those algorithms are at least exponential in the width of the input. An implementation of such algorithms would thus likely have to implement heuristics and optimizations not presented in the theory literature. For example, in~\cite{DBLP:conf/cp/FichteHZ19}, it was shown that one can use parallelism of GPUs to improve the efficiency of treewidth-based counting and thus scale to higher treewidth.

To close the paper, let us discuss several questions. First, the number of clauses of the encodings guaranteed by Theorem~\ref{thm:reverse} is very high. In particular, it is exponential in the width $k$. It would be interesting to understand if this can be avoided, i.e., if there are encodings of roughly the same primal treewidth whose size is polynomial in $k$.

It would also be interesting to see if our results can be extended to other classes of CNF-formulas on which SAT is tractable. Interesting classes to consider would e.g.~be the classes in~\cite{GanianS17}. In this paper, the authors define another graph for CNF-formulas for which bounded treewidth yields tractable model counting. It is not clear if the classes characterized that way allow small complete structured DNNF so our framework does not apply directly. It would still be interesting to see if one can show similar expressivity results to those here. Other interesting classes one could consider are those defined by backdoors, see e.g.~\cite{GaspersS13}. 

\paragraph*{Acknowledgements.}
The authors are grateful to the anonymous reviewers for their comments, which greatly helped to improve the presentation of the paper. 
The first author would like to thank David Mitchell for asking the right question at the right moment. This paper grew largely out of an answer to this question.

\bibliographystyle{abbrv}
\bibliography{twencodings}

\end{document}